\documentclass[11pt]{article}
\usepackage[letterpaper,margin=1in]{geometry}
\usepackage[numbers,sectionbib]{natbib}
\usepackage{times}
\usepackage{amsmath,amssymb,amsthm} 
\usepackage{graphicx}
\usepackage{multicol}
\usepackage{wrapfig}

\usepackage{algorithm}
\usepackage{algpseudocode}

\input{xy}
\xyoption{all}

\newtheorem{theorem}{Theorem}[section]
\newtheorem{lemma}[theorem]{Lemma}
\newtheorem{corollary}[theorem]{Corollary}

\newtheorem{claim}[theorem]{Claim}

\newcommand{\codespace}{\vspace{3mm}}

\newcommand{\ack}{F_{ack}}
\newcommand{\cack}{\cdot F_{ack}}

\newcommand{\id}[1]{\text{id}_{#1}}

\def\paragraph#1{\vspace{0.05em}\noindent {\bf #1}}
\def\paragraph#1{\vspace{0.25em}\noindent {\bf #1}}

\begin{document}

  \title{Consensus with an Abstract MAC Layer}
  \author{Calvin Newport\\ Georgetown University\\ {\tt cnewport@cs.georgetown.edu}}
 \date{}
  \maketitle
  
 \begin{abstract}
In this paper, we study distributed consensus in the radio network setting.
We produce new upper and lower bounds for this
problem in an abstract MAC layer model that captures the key guarantees provided by most wireless MAC layers.
In more detail, we first generalize the well-known impossibility of deterministic consensus with a single
crash failure~\cite{flp} from the asynchronous message passing model to our wireless setting.
Proceeding under the assumption of no faults, we then investigate
the amount of network knowledge required to solve consensus in our model---an important
question given that these networks are often deployed in an ad hoc manner.
We  prove consensus is impossible
without unique ids or without knowledge of network size (in multihop topologies).
We also prove a lower bound on optimal time complexity.
We then match these lower bounds with a pair of new deterministic consensus algorithms---one
 for single hop topologies and one for multihop topologies---providing a comprehensive
 characterization of the consensus problem in the wireless setting.
From a theoretical perspective, our results shed new insight into the role of network information and the power of MAC layer abstractions
in solving distributed consensus. From a practical perspective, given the level of abstraction used by our model,
our upper bounds can be easily implemented in real wireless devices on existing MAC layers
while preserving their correctness guarantees---facilitating the development of wireless distributed systems.  
 \end{abstract}
  
 
  




\section{Introduction}
\label{sec:intro}

Consensus provides a fundamental building block for developing reliable distributed
systems~\cite{guerraoui:1997,guerraoui:2000,guerraoui:2001}.
Motivated by the increasing interest in {\em wireless} distributed systems,
in this paper we prove new upper and lower bounds for the consensus problem in wireless networks.

\paragraph{The Abstract MAC Layer.}
Consensus bounds are dependent on the model in which they are established.
Accordingly, we must take care in selecting our model for studying the wireless version of this problem.
Most existing work on distributed algorithms for wireless networks assumes low-level synchronous models that require algorithms to deal directly
with link-layer issues such as signal fading and channel contention. Some of these models use topology graphs
to determine message behavior 
(c.f.,~\cite{baryehuda:1987,jurdzinski:2002,kowalski:2005,moscibroda:2005,czumaj:2006,gasieniec:2007})
while others use signal strength calculations 
(c.f.,~\cite{moscibroda:2006,moscibroda:2007,goussevskaia:2009,halldorsson:2012b,jurdzinski:2013:random,daum:2013}). 
These models are well-suited for asking basic science questions about the capabilities of wireless communication.
They are not necessarily appropriate, however, for developing algorithms meant for deployment,
as
  real wireless systems typically require an algorithm to operate on top of a general-purpose MAC layer which is hard to bypass and enables
many key network functions such as managing co-existence.

Motivated by this reality, in this paper we adopt the {\em abstract MAC layer} approach~\cite{kuhn:2011abstract},
in which we model the basic guarantees provided by most existing wireless MAC 
layers---if you broadcast a message
it will eventually be delivered with acknowledgment to nearby nodes in the network---but leverages a non-deterministic message
scheduler to allow for unpredictability---there is no bound on when messages are delivered or in what order. 
The goal with this approach is to describe and analyze algorithms at a level of abstraction
that makes it easy to subsequently implement theory results in real systems while still preserving their formally analyzed properties.
(See Section~\ref{sec:model} for a detailed model definition and motivation.)
%
%

\paragraph{Results.}
We begin with lower bounds.
In Section~\ref{sec:lower:crash}, we generalize the oft-cited result on the impossibility of deterministic consensus with a single
process failure~\cite{flp} from the asynchronous message passing model to our abstract MAC layer model.
(See Section~\ref{sec:model} for details on how these two models differ.)
The main difficulty in this generalization is the new assumption in our model that senders receive
acknowledgments at some point after a given broadcast completes. 
To overcome this difficulty, we are forced to restrict our valency definitions to focus on a restricted
class of schedulers.

Having established this impossibility, we proceed in this paper assuming no crash failures.
Noting that wireless network deployments are often {ad hoc},
we next focus on determining how much {\em a priori} information about the network
is required to solve deterministic consensus in our model.
We start, in Section~\ref{sec:lower:unique}, by proving that consensus is impossible without unique ids,
even if nodes know the size and diameter of the network. 
We then prove, in Section~\ref{sec:lower:n},
that even with unique ids (and knowledge of the diameter), consensus is impossible in multihop networks
if nodes do not know the network
size. Finally, we prove that any solution to consensus in our model requires
$\Omega(D\cack)$ time, where $D$ is the diameter of the underlying network topology
and $\ack$ is the maximum message delivery delay (a value unknown to the nodes in the network). 
All three bounds leverage partitioning arguments
that rely on carefully-constructed
worst-case network topologies and message scheduler behavior for
 which the relevant network knowledge assumptions do not break symmetry.

We then turn our attention to matching these lower bounds with a pair of new deterministic consensus algorithms.
We begin, in Section~\ref{sec:upper:single},
 with a new algorithm that guarantees to solve consensus in 
single hop networks in an optimal $O(\ack)$ time, even without advance
knowledge of the network size or participants (this opens up a gap with the asynchronous message
passing model, where consensus is impossible under such assumptions~\cite{abboud:2008}).
This algorithm uses a two-phase structure. The key insight is that nodes wait to decide after
their second phase broadcast until they have also heard this broadcast from a set of important {\em witnesses}.

We then present, in Section~\ref{sec:upper:multihop}, the {\em wireless PAXOS} ({wPAXOS}) algorithm, which 
guarantees to solve consensus in multihop topologies of diameter $D$ in an optimal
$O(D\cack)$ time. This algorithm assumes unique ids and knowledge of $n$ (as required by our lower bounds\footnote{Our algorithm
still works even if provided only {\em good enough} knowledge of $n$ to recognize a majority. This does not contradict
our lower bound as the lower bound assumes {\em no} knowledge of $n$.}),
but no other advance knowledge of the network or participants. 
The wPAXOS algorithm
combines the high-level logic of the PAXOS consensus algorithm~\cite{paxos}
with a collection of support services that efficiently disseminate proposals and aggregate responses.
We note that if the PAXOS (or similar consensus algorithm) logic is combined with a basic flooding algorithm,
the result would be a $O(n\cack)$ time complexity, as bottlenecks are possible where
$\Omega(n)$ value and id pairs must be sent by a single node only able to fit $O(1)$ such pairs in each message.
To reduce this time complexity to an optimal $O(D\cack)$, 
we implement eventually stable shortest-path routing trees and show they allow fast aggregation once stabilized,
and preserve safety at all times. These stabilizing support services and their analysis represent the main contribution of this algorithm.
One could, for example, replace the PAXOS logic working with these services with something simpler (since
we have unique ids and knowledge of $n$, and no crash failures,
we could, for example, simply gather all values at all nodes). We choose PAXOS mainly for performance reasons, 
as it only depends on a majority nodes to make progress, and is therefore not slowed if a small portion of the 
network is delayed.

\paragraph{Related Work.}
Consensus provides a fundamental building block for reliable distributed computing~\cite{guerraoui:1997,guerraoui:2000,guerraoui:2001}.
It is particularly well-studied in asynchronous models~\cite{paxos,schiper:1997,mostefaoui:1999,aguilera:2000},
where deterministic solutions are impossible with even a single crash failure~\cite{flp}.
Most existing distributed algorithm results for the wireless setting assume low-level models.
Though consensus has been studied in such models (e.g.,~\cite{chockler:2005}), most
efforts in the low-level setting focus on reliable communication problems such as broadcast
(see~\cite{peleg:2007} for a good survey).
The abstract MAC layer approach to modeling wireless networks is introduced in~\cite{kuhn:2009}
(later expanded to a journal version~\cite{kuhn:2011abstract}), and
has been subsequently used to study several
 different problems~\cite{cornejo2009neighbor,khabbazian:2010,khabbazian:2011,cornejo2014reliable}. 
This paper, however,  is the first to consider consensus in the 
abstract MAC layer context.

Other researchers have also studied consensus in wireless networks at higher levels of abstraction.
%
Vollset and Ezhilchelvan~\cite{vollset:2005}, 
and Alekeish and Ezhilchelvan~\cite{alekeish:2012}, study consensus
in a variant of the asynchronous message passing model where pairwise channels come and go dynamically---capturing 
some behavior of {mobile} wireless networks. Their correctness results depend on detailed liveness guarantees
that bound the allowable channel changes.
Wu et~al.~\cite{wu:2009} use the standard asynchronous message passing model (with 
unreliable failure detectors~\cite{chandra:1996}) as a stand-in for a wireless network,
focusing on how to reduce message complexity (an important metric in a resource-bounded wireless setting)
in solving consensus.

Finally, we note that a key focus in this paper is understanding the importance of network information
in solving consensus, a topic previously studied in the classical models.
Ruppert~\cite{ruppert2007anonymous},
and Bonnet and Raynal~\cite{bonnet2010anonymous},
for example, study the amount of extra power needed (in terms of shared objects and failure detection, respectively)
to solve wait-free consensus in {\em anonymous} versions of the standard models.
Attiya et~al.~\cite{attiya2002computing} describe consensus solutions for shared memory systems without failures or unique ids.
In this paper, by contrast, we prove consensus impossible without failures or unique ids. These results do not contradict,
however, as we assume multihop message passing-style networks. 
A series of papers~\cite{cavin:2004,greve:2007,alchieri:2008}, starting with the work of Cavin et~al.~\cite{cavin:2004},
study the related problem of  {\em consensus with unknown participants} (CUPs),
where nodes are only allowed to communicate with other nodes whose identities have been provided
by a {\em participant detector} formalism. Results on the CUPs problem focus on the structure of the knowledge
from such detectors required for consensus (e.g., if we create a graph with a directed edge indicating
participant knowledge, then the resulting graph must satisfy certain connectivity properties).
Closer to our own model is the work of Abboud~et~al.~\cite{abboud:2008},
which studies single hop networks in which participants are {\em a priori} unknown,
but nodes do have a reliable broadcast primitive. They prove consensus is impossible
in single hop networks under these assumptions without knowledge of network size.
In Section~\ref{sec:upper:single}, we describe an algorithm in our model that {\em does} solve consensus
under these assumptions: opening a gap between these two models.

\section{Model and Problem}
\label{sec:model}

For simplicity, in the following we sometimes call our model {\em the} {abstract MAC layer} model.
We emphasize, however, that there is no single abstract MAC layer model, but instead
many variants that share the same basic assumptions of acknowledged local broadcast
and an arbitrary scheduler.
%
%
%
The major differences between our model and the standard asynchronous message passing model 
are that: (1) we assume local broadcast instead of point-to-point communication;
(2) senders receive an acknowledgment at some point after their broadcast completes
(this acknowledgment captures the time at which the underlying link layer is done broadcasting its current message;
e.g., after its slot in a TDMA schedule arrives or its CSMA algorithm finally detected a clear channel);
and (3) we care about assumptions regarding network information knowledge as wireless networks
are often deployed in an ad hoc manner where such information may be unknown to nodes.

\paragraph{Model Details.}
To formalize our abstract MAC layer model,
fix a graph $G=(V,E)$, with the set $V$ describing the $|V|=n$ wireless devices in the network
(called {\em nodes} in the following), and the edges in $E$ describing nodes within reliable communication range.
In this model, nodes communicate with a local reliable (but not necessarily atomic\footnote{Local broadcast in wireless networks
is not necessarily an atomic operation, as effects such as the hidden terminal problem
might lead some neighbors of a sender to receive a message before other neighbors.})
broadcast primitive that guarantees to eventually deliver messages to a node's neighbors
in $G$.
At some point after a broadcast completes (see below), a node receives an ack. 
If a node attempts to broadcast additional messages before receiving an ack for the current message,
those extra messages are discarded.
To formalize the message delivery guarantees,
fix some execution $\alpha$ of a deterministic algorithm in our model.
To simplify definitions, assume w.l.o.g. that messages are unique.
Let $\pi$ be the event in $\alpha$ where $u$ calls {\em broadcast}$(m)$, and $\pi'$ be the subsequent ack returned to $u$.
Our abstract MAC layer model guarantees that in the interval from $\pi$ to $\pi'$ in $\alpha$, every non-faulty neighbor of $u$ in $G$ receives $m$,
and these are the only receive events for $m$ in $\alpha$ (this is where we leverage message uniqueness in our definition).

We associate each message scheduler with an unknown (to the nodes) but finite value $\ack$ that bounds the maximum
delay it is allowed between a broadcast and a corresponding acknowledgment.
This property induces some notion of fairness: the scheduler must eventually allow each broadcast to finish. 
To simplify timing, we assume local non-communication steps take no time. That is, all non-determinism is captured
 in the message receive and ack scheduling.
 We note that in some definitions of abstract MAC layer models (see~\cite{kuhn:2011abstract}),
 a second timing parameter, $F_{prog}$, is introduced to bound the time for a node to receive {\em some}
 message when one or more neighbors are broadcasting.
 We omit this parameter in this study as it is used mainly for refining time complexity analysis,
 while we are concerned here more with safety properties. Refining our upper bound results in a model
 that includes this second parameter remains useful future work.
 We also note that some definitions of the abstract MAC layer assume a second topology graph consisting of {\em unreliable} links
 that sometimes deliver messages and sometimes do not.
We omit this second graph in this analysis, which strengthens our lower bounds.
Optimizing our multihop upper bound to work in the presence of such links, however, is left an open question.

In some results that follow, we consider {\em crash} failures (a node halts for the remainder of the execution).
The decision to crash a node and the timing of the crash is determined by the scheduler and can happen
in the middle of a broadcast (i.e., after some neighbors have received the message but not all).
We call a node {\em non-faulty} (equiv. {\em correct}) with respect to a given execution if it does not crash.
For our upper bounds, we restrict the message size to contain at most a constant number of unique ids.
For a given topology graph $G$, we use $D$ to describe its diameter.
Finally,  for integer $i>0$, let $[i]=\{1,2,...,i\}$.



\paragraph{The Consensus Problem.}
To better understand the power of our abstract MAC layer model we explore upper and lower bounds for the standard binary consensus problem.
In more detail, 
each node begins an execution with an initial value from $\{0,1\}$. 
Every node has the ability to perform a single irrevocable {\em decide} action for a value in $\{0,1\}$.
To solve consensus, an algorithm must guarantee the following three properties:
(1) {\em agreement}: no
two nodes decide different values; (2) {\em validity}: if a node decides value $v$,
 then some node had $v$ as its initial value;
and (3) {\em termination}: every non-faulty process eventually decides.
By focusing on binary consensus,
as oppose to the more general definition that assumes an arbitrary value set,
we strengthen the lower bounds that form the core of this paper.
Generalizing our upper bounds to the general case in an efficient manner
(e.g., a solution more efficient than agreeing on the bits of a general value,
one by one, using binary consensus) is non-trivial and remains an open problem.


\section{Lower Bounds}
\label{sec:lower}

We begin by exploring the fundamental limits of our abstract MAC layer model with respect to the consensus problem.
In Section~\ref{sec:upper}, we provide matching upper bounds. In the following, we defer some proofs to
the appendix for the sake of clarity and concision. 

\subsection{Consensus is Impossible with Crash Failures}
\label{sec:lower:crash}

In this section we prove consensus is impossible in our model
in the presence of even a single crash failure. To achieve the strongest possible bound we assume a clique topology.
Our proof generalizes the FLP~\cite{flp} result to hold in our stronger setting where nodes now have acknowledgments.\footnote{The version
of the FLP proof that we directly generalize here is the cleaned up and condensed version that appeared in
 the subsequent  textbook of Lynch~\cite{lynch:1996}.}

\paragraph{Preliminaries.}
For this proof, assume w.l.o.g. that 
nodes always send messages; i.e., on receiving an {ack} for their current message
they immediately begin sending a new message.
We define a {\em step} of a node $u$ to be either:
(a) a node $v\neq u$ receiving $u$'s current message; or
(b) $u$ receiving an {ack} for its current message (at which point its algorithm advances
to sending a new message).
We call a step of type (a) from above {\em valid} with respect to the execution so far
 if the node $v$ receiving $u$'s message
has not previously received that message {\em and} all non-crashed nodes smaller than $v$ (by some fixed but arbitrary
ordering) have already received $u$'s message.
We call a step of type (b) {\em valid} with respect to the execution so far
if every non-crashed neighbor of $u$ has received its current message in a previous step.
When we consider executions that consist only of valid steps we are, in effect, restricting our attention to a particular
type of well-behaved message scheduler. 

We call an execution fragment (equiv. prefix) $\alpha$ of a consensus algorithm
{\em bivalent} if there is some extension of valid
steps that leads to nodes deciding $0$,
and some extension of valid steps that leads to nodes deciding $1$.
By contrast, we call an execution $\alpha$ {\em univalent}
if every extension of valid steps from $\alpha$ that leads to a decision leads to the same decision.\footnote{Notice, 
not every extension need lead to a decision. If, for example, the extension does not
give steps to two or more nodes, than this is equivalent to two or more nodes crashing---a circumstance
for which we do not expect a $1$-fault tolerant algorithm to necessarily terminate.}
If this decision is $0$ (resp. $1$), we also say that $\alpha$ is  {\em $0$-valent} (resp. {\em $1$-valent}).
In the following, we use the notation $\alpha \cdot s$, for execution fragment $\alpha$ and step $s$,
to describe the extension of $\alpha$ by $s$.

\paragraph{Result.}
Fix some algorithm ${\cal A}$. Assume for the sake of contradiction
that ${\cal A}$ guarantees to solve consensus in this setting with up to $1$ crash failure.
%
%
The key to generalizing the FLP impossibility to our model is the following lemma,
which reproves the main argument of this classical result in a new way that leverages
our model-specific constraints.

\begin{lemma}
Fix some bivalent execution fragment $\alpha$ of ${\cal A}$ and some process $u$. There exists a finite extension
$\alpha'$ of $\alpha$ such that $\alpha'\cdot s_u$ is bivalent, where $s_u$ is a valid step of $u$ with respect to $\alpha'$.
\label{lem:2}
\end{lemma}
\begin{proof}
Assume for contradiction that this property does not hold for some $\alpha$ and $u$.
It follows that for every finite extension $\alpha'$ of $\alpha$ of valid steps, if we extend $\alpha'$ by a single additional valid
step of $u$, the execution is univalent.
Let $s_u$ be the next valid step for $u$ after $\alpha$ (by our definition of valid, $s_u$ is well-defined).
We start by considering $\alpha\cdot s_u$. By assumption, this fragment is univalent.
Assume, w.l.o.g. that $\alpha\cdot s_u$ is $0$-valent (the argument below is symmetric
for the case where $\alpha\cdot s_u$ is instead $1$-valent). 
Because $\alpha$ is bivalent, however,
there is some other extension $\alpha''$ consisting of valid steps that is $1$-valent.

We now move step by step in $\alpha''$, starting from $\alpha$, 
until our growing fragment becomes $1$-valent.
Assume there are $k\geq 1$ steps in this extension.
Label the $k$ intermediate execution fragments from $\alpha$ to the $\alpha''$: $\alpha_1, \alpha_2,...,\alpha_k$,
where $\alpha_k$ is the where the fragment becomes $1$-valent.
To simplify notation, let $\alpha_0 = \alpha$.
For $0< i< k$, we know $\alpha_i$ is bivalent, so, by our contradiction assumption, that step between $\alpha_{i-1}$ and $\alpha_i$
cannot be $s_u$.
Let $s^*$ be the step between $\alpha_{k-1}$ and $\alpha_k$.
(Notice that it {\em is} possible that $s^* = s_u$, as the execution is no longer bivalent after this final step.)

By our contradiction assumption, we know that for each $i \in \{0,...,k-1\}$, $\alpha_i \cdot s_u$ is univalent.
We also know that $\alpha_0\cdot s_u$ is $0$-valent.
It follows that there must exist some $\hat i \in \{0,...,k-1\}$, such
that $\alpha_{\hat i}\cdot s_u$ is $0$-valent and $\alpha_{\hat i} \cdot s_v \cdot s_u$ is $1$-valent,
where $s_v$ is the next valid step of some node $v\neq u$.
Notice this holds whether or not $s^* = s_u$ (if $s^* = s_u$ then $\alpha_k$ is a fragment ending with $s_u$
that we know to be $1$-valent, otherwise, $\alpha_k \cdot s_u$ is this fragment).
We have found a fragment $\alpha^*$, therefore,
where $\alpha^* \cdot s_u$ is $0$-valent but $\alpha^* \cdot s_v \cdot s_u$ is $1$-valent.
We now perform a case analysis on $s_v$ and $s_u$ to show that all possible cases lead to a contradiction.
In the following, to simplify notation,
 let $\beta_0 = \alpha^* \cdot s_u$ and $\beta_1 = \alpha^* \cdot s_v \cdot s_u$.

{\em Case $1$:} Both steps affect the same node $w$.
It is possible that $w$ is $u$ or $v$ (e.g., if $s_u$ is an acknowledgment and $s_v$ is $u$ receiving $v$'s message),
it is also possible that $w$ is not $u$ or $v$ (e.g., if $s_u$ and $s_v$ are both receives at some third node $w$).
We note that $w$ (and only $w$) can distinguish between $\beta_0$ and $\beta_1$. 
Imagine, however, that we extend $\beta_1$ such that every node {\em except for $w$} keeps taking valid steps.
 All non-$w$ nodes must eventually decide, as this is equivalent to a fair execution where $w$ crashes after $\beta_1$,
 and $w$ is the only node to crash---a setting where termination demands decision.
  By our valency assumption, these nodes must decide $1$.
 Now imagine that we extend $\beta_0$ with the exact same steps.
 For all non-$w$ nodes these two executions are indistinguishable, so they will once again decide $1$.
 We assumed, however, that $\beta_0$ was $0$-valent: a contradiction.

{\em Case $2$:} The steps affect two different nodes. 
In this case, it is clear that no node can distinguish between $\beta_0\cdot s_v$ and $\beta_1$.
We can, therefore, apply the same style of indistinguishability argument as in case $1$,
except in this case we can allow all nodes to continue to take steps.
\end{proof}

\noindent We now leverage Lemma~\ref{lem:2} to prove our main theorem. 

\begin{theorem}
There does not exist a deterministic algorithm that guarantees to solve consensus in a single hop network
in our abstract MAC layer model
with a single crash failure.
 \label{thm:singlehop}
\end{theorem}
\begin{proof}
Assume for contradiction such an algorithm exists.
Call it ${\cal A}$.
Using the standard argument
we first establish the existence of  a bivalent initial configuration of ${\cal A}$ (e.g., Lemma $2$ from~\cite{flp}).
Starting from this configuration, we keep applying Lemma~\ref{lem:2},
rotating through all $n$ nodes in round robin order,
to extend the execution in a way that keeps it bivalent.
Because we rotate through all nodes when applying Lemma~\ref{lem:2},
the resulting execution is fair in the sense that all nodes keep taking steps.
The termination property of consensus requires that nodes eventually decide.
By agreement when any node decides the execution becomes univalent.
By Lemma~\ref{lem:2}, however, our execution remains bivalent, so no node
must ever decide. This violates termination and therefore contradicts the assumption that ${\cal A}$
solves consensus.
\end{proof}



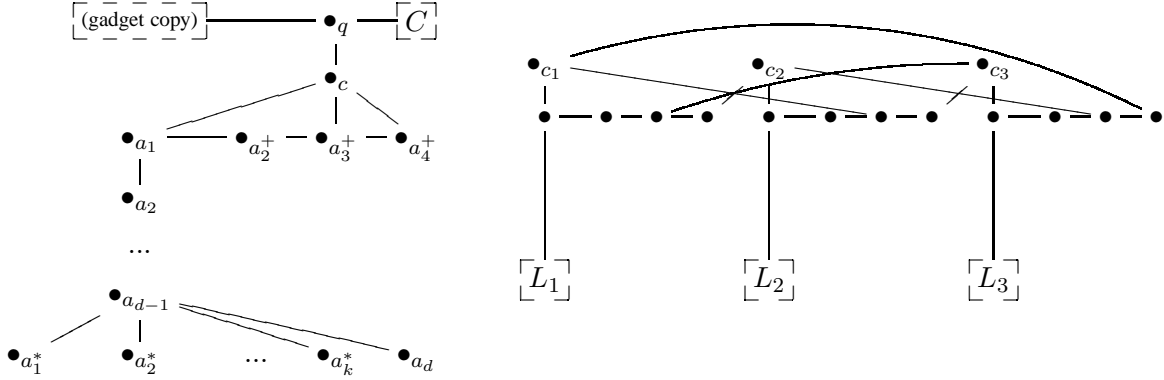
\begin{figure}[!t]
\centering
\begin{minipage}{.39\textwidth}
  \centering

\xymatrix@=.7em
{
                                 &   &   &                          &                                        &  \\
                                 &*+[F--]{\text{\scriptsize (gadget copy)}}    &   &   {\bullet}_{q} \ar@{-}[r] \ar@{-}[ll] & *+[F--]{C}   &  \\
                                 &   &   &   {\bullet}_{c} \ar@{-}[u]  &   &   \\
                                 & {\bullet} _{a_1}  \ar@{-}[r] \ar@{-}[urr] &  {\bullet}_{a^{+}_2} \ar@{-}[r] & {\bullet}_{a^{+}_3} \ar@{-}[r]  \ar@{-}[u]& {\bullet}_{a^{+}_4} \ar@{-}[ul] &   \\
                                 &  {\bullet} _{a_2} \ar@{-}[u]&   &   &   &   \\
                                 &  ... &   &   &   &   \\
                                 &  {\bullet}_{a_{d-1}} &   &   &   &   \\
  {\bullet}_{a^{*}_1} \ar@{-}[ur]  & {\bullet}_{a^{*}_2} \ar@{-}[u] & ...  & {\bullet}_{a^{*}_k} \ar@{-}[ull]& {\bullet}_{a_d}\ar@{-}[ulll] &   
 %
  }

\end{minipage}%
\begin{minipage}{.61\textwidth}
  \centering

\xymatrix@=.7em
{
 & {\bullet}_{c_1} & & & &{\bullet}_{c_2}  & & & & {\bullet}_{c_3} & & &\\
 & {\bullet} \ar@{-}[u] \ar@{-}[r] & {\bullet} \ar@{-}[r] & {\bullet}\ar@{-}[r]  \ar@{-}@/^/[urrrrrr]& {\bullet}\ar@{-}[ur]  & {\bullet}\ar@{-}[u]\ar@{-}[r]  
      & {\bullet}\ar@{-}[r] & {\bullet}\ar@{-}[r]  \ar@{-}[ullllll] & {\bullet}  \ar@{-}[ur] & {\bullet}\ar@{-}[u] \ar@{-}[r] & {\bullet} \ar@{-}[r] & {\bullet} \ar@{-}[r]  
        \ar@{-}[ullllll] & {\bullet} \ar@{-}@/_2pc/[ulllllllllll]\\
        & & & & & & & & & & & & \\
                & & & & & & & & & & & & \\
                        & & & & & & & & & & & & \\
  & *+[F--]{L_1} \ar@{-}[uuuu] & & & & *+[F--]{L_2} \ar@{-}[uuuu]  & & & & *+[F--]{L_3} \ar@{-}[uuuu]  & & & 
}

\end{minipage}
\caption{In {\bf Network A} (left) the $a$ nodes plus $c$ combine to comprise a {\em gadget}. The bridge node $q$ is connected to two copies
of this gadget at their $c$ nodes. It is also connected to all nodes in a clique $C$ used to adjust the total network size. In {\bf Network B} (right) the sub-graphs $L_1$, $L_2$, and $L_3$,
are each a copy of the sub-graph of the gadget of Network $A$ consisting of nodes at $a_2$ and below in the diagram
(i.e., nodes labelled $a_i$ for $i>1$, as well as the $a^*_j$ nodes.) $L_1$, $L_2$, and $L_3$, in other words, connect to the $a_1$ node
of the Network $A$ gadget.}
\label{fig:uid}
\end{figure}

\subsection{Consensus is Impossible without Unique Ids}
\label{sec:lower:unique}

Having proved that consensus is impossible with crash failures,
we consider the conditions under which it remains impossible {\em without} crash failures.
Recall, in wireless networks, unlike wired networks, the network configuration might be {ad hoc},
preventing nodes from having full {\em a priori} information on the participants.
Accordingly, in this section and the next we explore
the network information required to solve consensus.
We start here by investigating the importance of unique ids.
We call an algorithm that {\em does not} use unique ids an {\em anonymous} algorithm.
We prove below that consensus is impossible with anonymous algorithms, even if nodes know
the network size and diameter. 
We then provide a corollary that extends this result to the standard asynchronous network model.
To the best of our knowledge, this is the first result on the necessity of unique ids for consensus in multihop message passing networks.

\paragraph{Result.}
To prove our main theorem we leverage an indistinguishability result based on the network topologies shown
in Figure~\ref{fig:uid}. Due to the careful construction of these networks, we cannot prove our impossibility
holds for all $n$ and $D$ (network $B$, for example, requires that $n$ be divisible by $3$).
We can, however, prove that for every sufficiently large (even) $D$ and $n$, 
the problem is impossible for $D$ and some $n' = \Theta(n)$.

\begin{theorem}
There exists a constant integer $c\geq 1$,
such that for every
even diameter $D\geq 4$ and network size $n \geq D$,
there exists an $n' \in \{n,...,c\cdot n\}$,
such that no anonymous algorithm ${\cal A}$ guarantees to solve consensus in our abstract MAC layer 
model in all networks of
diameter $D$ and size $n'$.
\label{thm:nounique}
\end{theorem}

\noindent Given some $D$ and $n$ that satisfy the theorem constraints,
 let $k$ be the smallest integer $k\geq 0$ such that $3(\frac{D-2}{2} +k) + 12 \geq n$.
Set $n' = 3(\frac{D-2}{2} + k) + 12$.
Consider networks $A$ and $B$ from Figure~\ref{fig:uid},
instantiated with $d=\frac{D-2}{2}$ and $k$ set to the value fixed above in defining $n'$.
In the case of network $A$, set the clique $C$ to contain enough nodes
to bring the total count in that network to equal the number of nodes in $B$.
The following claim follows from the structure of these networks
and our definitions of $k$ and $d$.

\begin{claim}
Networks $A$ and $B$, instantiated with the values described above,
 have size $n'$ and diameter $D$.
\label{claim:nounique:1}
\end{claim}

\noindent We define the {\em synchronous scheduler} in our model to be a message scheduler
that delivers messages in lock step rounds. That is, it delivers all nodes' current message
to all recipients, then provides all nodes with an {ack}, and then moves on to the
next batch of messages.
Furthermore, we assume no global time (or, equivalently, some small amount of time that
we can define as needed in the below proof) passes between these {\em synchronous steps.}
Fix some consensus algorithm ${\cal A}$ that does not use unique ids. 
For $b\in \{0,1\}$, let $\alpha_B^b$ be the execution of ${\cal A}$ in network $B$ (see the right network
in Figure~\ref{fig:uid})
with all nodes starting with initial value $b$ and message behaviors
 scheduled by the synchronous scheduler.
The following lemma follows directly from the definition of consensus and the fairness of the synchronous scheduler.

\begin{lemma}
There exists some $t\geq 0$ such that for $b\in \{0,1\}$, $\alpha_B^b$ terminates
by synchronous step $t$ with all nodes deciding $b$.
\label{lem:nounique:1}
\end{lemma}

Next, let $\alpha_A$ be an execution of ${\cal A}$ in network
$A$ (see the left network in Figure~\ref{fig:uid}) defined as follows:
(1) all nodes in one gadget start with initial value $0$ (call these nodes $A_0$),
all nodes in the other copy of the gadget start with initial value $1$ (call these nodes $A_1$);
(2) the bridge node $q$ and the nodes in component $C$ start with arbitrary initial values;
and (3) we fix the scheduler to schedule the steps of $A_0$ and $A_1$ like the synchronous
scheduler for
for $t$ steps (for the $t$ fixed in Lemma~\ref{lem:nounique:1}),
while delaying any message from node $q$ being delivered
until after these $t$ steps are complete. After this point, the scheduler
can behave as the synchronous scheduler for the full network.

The key argument in our proof is that a node in $A_b$ cannot distinguish itself
during the first $t$ steps of $\alpha_A$ from the same node in $\alpha_B^b$.
Intuitively, this follows because the network in $B$ is carefully constructed
to be symmetric, 
so nodes cannot tell if they are communicating with one copy of the network $A$ gadget
or multiple copies.
To formalize this argument, we introduce some notion that relates network $A$ to $B$.
Notice that network $B$ consists of three copies of the gadget from network $A$
(with some of the edges
from the connector node copies swapped to interconnect the copies).
For each node $u$ in a network $A$ gadget, therefore,
we can define  $S_u$  to be the set containing the three
nodes in network $B$ that correspond to $u$: that is, the nodes in $u$'s position
in the three gadget copies of $B$).
For example, consider node $c$ in the network $A$ gadget shown in Figure~\ref{fig:uid}.
By our above definition,  $S_c = \{c_1, c_2, c_3\}$.
We can now formalize our indistinguishability.

\begin{lemma}
Fix some $b\in \{0,1\}$ and some node $u$ in $A_b$.
The first $t$ steps of $u$ in $\alpha_A$
are indistinguishable from the first $t$ steps of the three nodes in $S_u$
in $\alpha_B^b$.
\label{lem:nounique:2} 
\end{lemma}
\begin{proof}
%
We begin by noting the following property of our networks that follows directly from its structure
(in the following, we use the notation $N_{A_b}$ to indicate the neighbor function of the subgraph
of network $A$ consisting only of the nodes in $A_b$):
%
%
(*) {\em Fix any $u\in A_b$ and $u'\in S_u$. For every $v\in N_{A_b}(u)$, $u'$ is connected
to exactly one node in $S_v$. There are no other edges adjacent to $u'$ in $B$.}
%
We now leverage property (*) in proving the following induction argument, which itself
directly implies our lemma statement.
The below induction is on the number of synchronous steps in the $\alpha$ executions.
%

{\em Hypothesis:} $\forall u\in A_b$,  $0 \leq r \leq t$:
after $r$ steps,
$u$ in $\alpha_A$ has the same state as the nodes in $S_u$ in $\alpha_B^b$.

{\em Basis ($r=0$):} Because we assume no unique ids and the same initial values for all relevant nodes, 
 the hypothesis is trivially true after $0$ steps.

{\em Step:} Assume the hypothesis holds through some step $r, 0 \leq r < t$. We will now show it holds for step $r+1$.
By our hypothesis, for each $w\in A_b$, the nodes in $S_w$ will send the same message as $w$ during step $r+1$
(as this message is generated deterministically by the nodes' state after step $r$).
Now consider a particular $u\in A_b$ and a particular copy $u' \in S_u$ in network $B$.
By property (*), for each node $v\in N_{A_b}(u)$ that sends a message to $u$ in $r+1$,
$u'$ is connected to a single node in $S_v$. By our above argument, this node in $S_v$
will send the same message to $u'$ as $v$ sends to $u$. Furthermore, (*) establishes
that there are no other edges to $u'$ that will deliver messages at this point. It follows that $u'$ will receive the same message set
in $r+1$ in $\alpha_B^b$ as $u$ receives in $r+1$ in $\alpha_A$. They will end $r+1$,
therefore, in the same state.
\end{proof}

\noindent We now leverage Lemma~\ref{lem:nounique:2} to prove our main theorem.

\begin{proof}[Proof of Theorem~\ref{thm:nounique}.]
Assume for contradiction that there exists an anonymous algorithm ${\cal A}$ that guarantees
to solve consensus for a diameter $D$ and network size $n'$ specified to be impossible by the theorem statement.
Fix some nodes $u\in A_0$ and $v\in A_1$ such that $u$ and $v$ are in the same position in their respective
gadgets in network $A$.
Fix some $w\in S_u = S_v$.
By Lemma~\ref{lem:nounique:1}, $w$ decides $0$ within $t$ steps
of $\alpha_B^0$. Combining this observation
with Lemma~\ref{lem:nounique:2}, 
applied to $u$ and $b=0$, it follows that $u$ will decide $0$ in $\alpha_A$.
By agreement, it follows that all nodes must decide $0$ in $\alpha_A$---including $v$.
We can, however, apply this same argument to $\alpha_B^1$, $v$ and $b=1$, to 
determine that $v$ decides $1$ in $\alpha_A$. A contradiction.
 \end{proof}

\noindent We conclude with a corollary for the standard asynchronous model that follows
from the fact that our model is strictly stronger and this result concerns a lower bound.

\begin{corollary}
There exists a constant integer $c\geq 1$,
such that for every
even diameter $D\geq 4$ and network size $n \geq D$,
there exists an $n' \in \{n,...,c\cdot n\}$,
such that no anonymous algorithm ${\cal A}$ guarantees to solve consensus in the asynchronous network model
with broadcast communication and no advance knowledge of the network topology,
in all networks of
diameter $D$ and size $n'$.
\end{corollary}

\subsection{Consensus is Impossible without Knowledge of $n$}
\label{sec:lower:n}

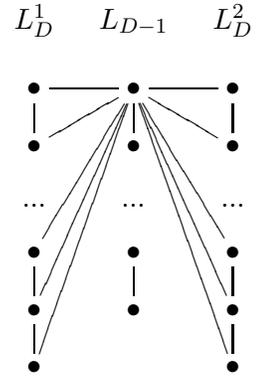
\begin{wrapfigure}{c}{0.3\textwidth}
%
\centerline{
\xymatrix @=1em
{
L_D^1   & L_{D-1}   &  L_D^2 \\
\bullet  \ar@{-}[d]   \ar@{-}[r] &   \bullet  \ar@{-}[d]     & \bullet  \ar@{-}[d]  \ar@{-}[l] \\
\bullet   \ar@{-}[ur]  &   \bullet     & \bullet  \ar@{-}[ul] \\
 ...  &  ...     &  ... \\
\bullet   \ar@{-}[d]   \ar@{-}[ruuu] &   \bullet   \ar@{-}[d]    & \bullet  \ar@{-}[d]  \ar@{-}[luuu]  \\
\bullet  \ar@{-}[d]  \ar@{-}[ruuuu] &   \bullet     & \bullet  \ar@{-}[d] \ar@{-}[luuuu]  \\
\bullet  \ar@{-}[ruuuuu] &       & \bullet \ar@{-}[luuuuu] 
}
}
\caption{The {$K_D$ network. Note: $L_{D-1}$ contains $D$ nodes.} 
%
}
\label{fig:kd} 
\end{wrapfigure}

In Section~\ref{sec:lower:unique}, we proved that consensus in our model requires unique ids.
Here we prove that even with unique ids and knowledge of $D$, nodes still need knowledge
of $n$ to solve the problem (in multihop networks). 
%
%
Our strategy for proving this theorem is an indistinguishability argument of a similar style
to that used in Section~\ref{sec:lower:unique}.
In more detail, consider network $K_D$ with diameter $D$ shown in Figure~\ref{fig:kd}.
Imagine that we start the $D+1$ nodes in sub-graph $L_D^1$ (resp. $L_D^2$) with initial value $0$ (resp. $1$).
If we delay message delivery long enough between $L_{D-1}$ and its neighbors in $K_D$,
the nodes in $L_D^i$ cannot distinguish between being partitioned in $K_D$
or executing by themselves in a network. Diameter knowledge does not distinguish these cases.

To formalize this argument,
we first assume w.l.o.g.
that nodes continually send messages.
In the following, fix some consensus algorithm ${\cal A}$.
Let $L_d$, for integer $d\geq 1$, be the network graph consisting of $d+1$ nodes in a line.
Let $\alpha^b_d$, for $b\in \{0,1\}$ and some integer $d\geq 1$, 
be the execution of ${\cal A}$ in $L_d$, where all nodes begin with initial value $0$
and message behavior is scheduled by the synchronous scheduler (defined in Section~\ref{sec:lower:unique}).
The following lemma follows from the validity and termination properties of consensus.

\begin{lemma}
There exists some integer $t\geq 0$, such
that for every $b\in \{0,1\}$ and $d\geq 1$, $\alpha^b_d$ terminates after $t$ synchronous steps with all nodes deciding $b$.
\label{lem:uniqueanddiameter:1}
\end{lemma}

\noindent For a given diameter $D>1$, we define 
the network graph $K_D$ 
to consist of two copies of $L_D$ (call these $L_D^1$ and $L_D^2$) and the line $L_{D-1}$, with an edge added from every node
in $L_D^1$ and $L_D^2$ to some fixed endpoint of the $L_{D-1}$ line.
Notice that by construction, $K_D$ has diameter $D$. (See Figure~\ref{fig:kd}.)
Next, 
we 
define the {\em semi-synchronous scheduler}, in the context of network graph $K_D$,
to be a message scheduler that delivers messages amongst nodes in $L_D^1$ and amongst nodes in $L_D^2$,
in the same manner as the synchronous scheduler for $t$ synchronous steps (for the $t$ provided by Lemma~\ref{lem:uniqueanddiameter:1}).
During this period, the semi-synchronous scheduler does {\em not} deliver any messages from the endpoint of the $L_{D-1}$
line to nodes in $L_D^1$ or $L_D^2$. After this period, it behaves the same as the synchronous scheduler.
Let $\beta_D$ be the execution of ${\cal A}$ in $K_D$ with: (1) all nodes in $L_D^1$ starting with initial value $0$;
(2) all nodes in $L_D^2$ starting with initial value $1$; (3) all nodes in $L_{D-1}$ starting with arbitrary initial values;
and (4) the semi-synchronous scheduler controlling message actions.
With these definitions established, we can prove our main theorem.

\begin{theorem}
For every $D>1$, no algorithm ${\cal A}$ guarantees to
solve consensus in our abstract MAC layer model
in all networks of diameter $D$.
\label{thm:uniqueanddiameter}
\end{theorem}
\begin{proof}
Assume for contradiction that ${\cal A}$ guarantees to solve consensus in all networks
of diameter $D$, for some fixed $D>1$.
By the definition of the semi-synchronous scheduler,
it is straightforward to see that
$\beta_D$ is indistinguishable from $\alpha^0_D$ for nodes in $L_D^1$,
and indistinguishable from $\alpha^1_D$ for nodes in $L_D^2$,
for the first $t$ synchronous steps.
Combining Lemma~\ref{lem:uniqueanddiameter:1} with our indistinguishability,
we note that nodes in $L_D^1$ will
decide $0$ in $\beta_D$ while nodes in $L_D^2$ will decide $1$.
Therefore, ${\cal A}$ does not satisfy agreement in $\beta_D$.
We constructed $K_D$, however, so that it has a diameter of $D$.
Therefore, ${\cal A}$ guarantees to solve consensus (and thus satisfy agreement) in this network. A contradiction.
\end{proof}

\subsection{Consensus Requires $\Omega(D\cack)$ Time}
\label{sec:lower:time}

The preceding lower bounds all concerned computability.
For the sake of completeness, we conclude by considering complexity.
The $\Omega(D\cack)$ time bound claimed below
is established by a partitioning argument.

\begin{theorem}
No algorithm can guarantee to solve consensus in our abstract MAC layer model
in less than $\lfloor \frac{D}{2} \rfloor \ack$ time.
\label{thm:time}
\end{theorem} 
\begin{proof}
Fix some $D$. Consider a line of diameter $D$ consisting of nodes $u_1,u_2,...,u_{D+1}$, arranged in that order.
Consider an execution of a consensus algorithm in this network with a variant of the synchronous scheduler
from Section~\ref{sec:lower:unique} that delays the maximum $\ack$ time between each synchronous step.
In $\lfloor \frac{D}{2} \rfloor \ack$ time, the endpoints cannot hear from beyond their nearest half of the line.
If we assume they must decide by this deadline we can apply a standard partitioning argument to create an agreement violation.
In particular, if one half starts with initial value $0$ and the other initial value $1$,
the endpoint of the first half must decide $0$ and the endpoint of the second must decide $1$ (by indistinguishability and validity),
creating an agreement violation. 
\end{proof}

\section{Upper Bounds}
\label{sec:upper}

In Section~\ref{sec:lower}, we proved fundamental limits
on solving consensus in our model.
In this section, we prove these bounds optimal with matching upper bounds.
We consider both {\em single hop} (i.e., the network graph is a clique)
and {\em multihop} (i.e., the network graph is an arbitrary connected graph) networks.
Due to the impossibility result from Section~\ref{sec:lower:crash}, we assume no crash failures
in the following.

\subsection{Consensus in Single Hop Networks}
\label{sec:upper:single}

Here we describe an algorithm---{\em two-phase consensus}---that guarantees to solve
consensus in single hop network topologies in an optimal $O(\ack)$ time.
It assumes unique ids but does not require knowledge of $n$.
This opens a separation with the standard broadcast asynchronous model (which does
not include acknowledgments) where
consensus is known to be impossible under these conditions~\cite{abboud:2008}.
%
%
%
%
\begin{wrapfigure}{L}{0.5\textwidth}
\begin{minipage}{0.5\textwidth}
\begin{algorithm}[H]
  \caption{Two-Phase Consensus (for node $u$)}      
       \begin{algorithmic}[1]
  \scriptsize
 \State $v \gets$ initial value from $\{0,1\}$
 \State $R_1 \gets \{\langle \text{phase 1}, \id{u}, v \rangle\}$  
 \State \Comment{Phase $1$}
 \State {\bf broadcast}$(\langle \text{phase 1}, \id{u}, v \rangle)$  
 \While{waiting for $ack$}
 \State add {\bf received} messages to $R_1$
 \EndWhile
 \If{$\langle \text{phase 1,*,} 1-v \rangle\in R_1$ or $\langle \text{phase 2,*, bivalent} \rangle\in R_1$}
 	\State $status\gets$ bivalent
\Else
	\State $status \gets$ decided$(v)$
\EndIf
\State \Comment{Phase 2}
\State {\bf broadcast}$(\langle  \text{phase 2}, \id{u}, status \rangle)$
\State $R_2 \gets \{\langle  \text{phase 2}, \id{u}, status \rangle\}$
\While{waiting for $ack$}
     \State add {\bf received}  messages to $R_2$
\EndWhile
\State $W \gets$ every unique id in $R_1$ and $R_2$ \Comment{Witness list created}
\While{$\exists id\in W$ s.t. $\langle \text{phase 2, id,*} \rangle \notin R_1 \cup R_2$}
     \State add {\bf received} phase $2$ messages to $R_2$
\EndWhile
\If{$\langle \text{phase 2, *, decided}(0) \rangle \in R_2$}
  \State {\bf decide} 0
\Else
\State {\bf decide} 1
\EndIf
   \end{algorithmic}

       \label{alg:singlehop} 
      \end{algorithm}
\end{minipage}
\end{wrapfigure}

%
%
%
%
%
%
%
The pseudocode for our algorithm is presented in Algorithm~\ref{alg:singlehop}.
Here we summarize its operation:
Each node $u$ executes two {\em phases}.
At the beginning of the first phase, $u$ broadcasts its unique id and its initial value $v_u\in \{0,1\}$.
Node $u$ considers its first phase complete once it receives an acknowledgment for its first broadcast.
At this point, $u$ will choose its {\em status}.
If $u$ has seen evidence of a different initial value in the system by this point (i.e., it sees a phase $1$ message
for a different value or a {\em bivalent} phase $2$ message), 
it sets its status to {\em bivalent}.
Otherwise, it sets it to {\em decided $v_u$}.
Node $u$ now begins phase $2$ by broadcasting its status and id.
Once $u$ finishes this phase $2$ broadcast, it has two possibilities.
If its status is {\em decided}, then it can decide its initial value and terminate.
Otherwise,
 it constructs a {\em witness set} $W_u$,
consisting of every node it has heard from so far in the execution. 
It waits until it has received a phase $2$ message from {\em every} node
in $W_u$.
At this point, if the set contains any message of the form {\em decided $v_w$},
then it decides $v_w$.
Otherwise, it decides default value $1$.
We now establish the correctness of this strategy.

 
\begin{theorem}
The two-phase consensus algorithm solves consensus in
 $O(\ack)$ time
in our abstract MAC layer model in single hop networks with unique ids.
\label{thm:single:upper}
\end{theorem}
\begin{proof}
Validity and termination are straightforward to establish.
We turn our attention, therefore, to agreement.
%
If no node ends up with 
{\em status = decided$(0)$}, then $1$ is the only possible decision value.
The interesting case, therefore, is when some node $u$ does set  {\em status $\gets$ decided$(0)$}
after its phase $1$ broadcast completes.
Let $S$ be the subset of nodes that began with initial value $1$ (if any). 
By assumption, $u$'s phase $1$ broadcast completed before any node in $S$,
as, otherwise, $u$ would have seen evidence of a $1$ before setting $status$,
preventing it from choosing {\em decided}$(0)$.
It follows that every node in $S$ must set $status$ to {\em bivalent.}
We know, therefore, that it is impossible to have both {\em decided}$(1)$ and {\em decided}$(0)$
in the system. We are left to show that if there is {\em decided}$(0)$ in the system, then all nodes 
end up deciding $0$.
As before, let $u$ be a node with status {\em decided}$(0)$.
Now let $v$ be a node with status {\em bivalent}.
We consider two cases concerning $u$ and $v$'s interaction.

In the first case, assume $v$ receives a message from $u$ before $v$ finishes its phase $2$ broadcast.
It follows that $u$ will be placed in $v$'s witness list, $W$.
The algorithm now requires $v$ to wait for $u$'s phase $2$ broadcast before deciding.
It will therefore see that $u$ has a status of {\em decided}$(0)$, requiring $v$ to decide $0$.
In the second case, $v$ does not receive a message from $u$ before $v$ finishes
its phase $2$ broadcast. Accordingly, $u$ is not in $v$'s witness set $W$.
This might be problematic as it could allow $v$ to decide before it sees a {\em decided}$(0)$ message.
Fortunately, we can show this second case cannot happen.
If $v$ had not heard {\em any} message from $u$ by the time it finished its phase $2$ broadcast,
it follows that $u$ receives this broadcast before it finishes its phase $1$ broadcast.
But $v$'s phase $2$ broadcast has a {\em bivalent} status. By the algorithm, this would 
prevent $u$ from setting its status to {\em decided}$(0)$---contradicting our assumption that $u$ has
a {\em decided} status.
\end{proof}



\subsection{Consensus in Multihop Networks}
\label{sec:upper:multihop}

We now describe a consensus algorithm for the multihop setting
that guarantees to solve consensus in $O(D\cack)$ time.
It assumes unique ids and knowledge of $n$ (as required by the lower bounds
of Section~\ref{sec:lower}), but makes no additional assumptions about the participants or network
topology.
Notice, this solution does not {\em replace} the single hop algorithm of Section~\ref{sec:upper:single},
as this previous algorithm: (1) is simpler; (2) has a small constant in its time complexity (i.e., $2$);
and (3) does not require knowledge of $n$.\footnote{This lack of knowledge of $n$ does
not violate the lower bound of Section~\ref{sec:lower:n}, as this lower bound requires the use
of a multihop network topology.}


Our strategy for solving consensus in this setting is to leverage the logic
of the PAXOS consensus algorithm~\cite{paxos,paxos-simple}. 
This algorithm was designed and analyzed for the asynchronous network model
with bounded crash failures.
Here we apply the logic to our wireless model with no crash failures.
The main difficulty we face in this effort is that nodes do not know the topology of the network
or the identity of the other participants in advance.
To overcome these issue we connect the PAXOS logic with a collection of sub-routines we
call {\em services}, which are responsible for efficiently delivering messages, electing the leaders
needed by PAXOS for liveness, and telling the proposers when to generate new proposal numbers.
We call this combination of PAXOS logic with our model-specific services
  {\em wireless PAXOS} (wPAXOS).

If we were satisfied with a non-optimal $O(n\cack)$ time complexity, the communication services
could be implemented with a simple flooding logic (the first time you see a message, re-broadcast),
and the leader election service could simply return the largest id seen so far.
To obtain an optimal $O(D\cack)$ time complexity, however, requires a more intricate solution.
In particular, when a proposer is waiting to hear from a majority of acceptors, 
we cannot afford for it to receive each response individually (as each message can only hold a constant number
of unique ids, and this would therefore require a proposer to receive $\Theta(n)$ messages).
Our solution is to instead  have nodes execute a distributed Bellman-Ford style iterative refinement strategy to establish
shortest-path routing trees rooted at potential leaders.
We design this service such that once the leader election service stabilizes,
a tree rooted at this leader will complete soon after (if it is not already completed).
These trees are then used to safely {\em aggregate} responses from acceptors: a strategy
that leverages the fact that PAXOS only requires the total {\em count} of a given response type,
not the individual responses.\footnote{A technicality here is that these responses sometimes include
prior proposals; we handle this issue by simply maintaining in aggregated responses
the prior proposal---if any---with the largest proposal number of those being aggregated.}
This aggregation strategy reduces the time to gather responses (after stabilization)
from $O(n\cack)$ to $O(D\cack)$. 

The final optimization needed to adapt PAXOS to our model is the change service.
We need the eventual leader to generate proposals {\em after} the leader election and tree services
stabilize, so it can reap the benefits of efficient acceptor response aggregation.
At the same time, however, the leader cannot generate {\em too many} new proposals
after this point, or each new proposal may delay the previous. The key property
of our change service is that it guarantees that the leader will generate $\Theta(1)$
new proposal after stabilization (assuming there is no decision yet in the network).
%



\begin{figure}

\noindent\begin{minipage}{0.48\textwidth}

\begin{algorithm}[H]
  \caption{Leader Election Service (for node $u$)}
  \label{alg:paxos:leader}
  \begin{algorithmic}[1]
  \scriptsize
  \Procedure {On Initialization}{}
    \State $\Omega_u \gets \id{u}$
    \State {UpdateQ}$(\langle leader, \id{u}\rangle)$
  \EndProcedure  
   \codespace
    \Procedure{Receive}{$\langle leader, id \rangle$}
    \If{$id > \Omega_u$}
    \State $\Omega_u \gets id$
     \State {UpdateQ}$(\langle leader, id\rangle)$
   \EndIf
   \EndProcedure
 \codespace  
 \Procedure {UpdateQ}{$\langle leader, id \rangle$}
\State {\bf empty} leader queue and {\bf enqueue} $\langle leader, id \rangle$ 
\EndProcedure
   \end{algorithmic}
\end{algorithm}
\begin{algorithm}[H]
  \caption{Change Service (for node $u$)}
  \label{alg:paxos:change}
  \begin{algorithmic}[1]
  \scriptsize
  \Procedure {On Initialization}{}
    \State $lastChange \gets -\infty$
  \EndProcedure  
   \codespace
  \Procedure {OnChange}{} \Comment{$\Omega_u$ or $dist_u$ updated.}
    \State $lastChange \gets$ time\_stamp() 
    \State UpdateQ$(\langle change, lastChange,\id{u} \rangle)$
  \EndProcedure
  \codespace
   \Procedure{Receive}{$\langle change, t, id \rangle$}
    \If{$t>lastChange$}
    \State $lastChange \gets t$
     \State {UpdateQ}$(\langle change, t, id\rangle)$
   \EndIf
   \EndProcedure
    \codespace
 \Procedure {UpdateQ}{$\langle change, t, id \rangle$}
 \State {\bf empty} the change queue then {\bf enqueue} $\langle change, t, id \rangle$
  \If{$\Omega_u = \id{u}$} \State GenerateNewPAXOSProposal()  \EndIf
      \EndProcedure  
   \end{algorithmic}
\end{algorithm}
\end{minipage}
%
\noindent\begin{minipage}{0.48\textwidth}
\begin{algorithm}[H]
  \caption{Tree Building Service (for node $u$)}
  \label{alg:tree}
  \begin{algorithmic}[1]
   \scriptsize
 \Procedure {On Initialization}{}
    \State $\forall v\neq u: dist[\id{v}] \gets \infty$ and $parent[\id{v}] \gets \bot$\
    \State $dist[\id{u}] \gets 0$ and $parent[\id{u}] \gets \id{u}$
    \State UpdateQ$(\langle search, \id{u}, 1\rangle)$
      \EndProcedure  
   \codespace
    \Procedure{Receive}{$m = \langle search, id, h \rangle$}
    \If{$h < dist[id]$}
    \State $dist[id] \gets h$
    \State $parent[id] \gets m.sender$
    \State {UpdateQ}$(\langle search, id, h+1 \rangle)$
   \EndIf
   \EndProcedure
 \codespace  
 \Procedure {UpdateQ}{$\langle search, id, h \rangle$}
 \State {\bf enqueue} $\langle search, id, h \rangle$ on tree queue
 \State {\bf discard} any message for $id$ with hop count $h' > h$ 
 \State {\bf move} message (if any) with id  $\Omega_u$ to front of tree queue
 \EndProcedure
 \codespace
 \Procedure{OnLeaderChange}{} \Comment{Called when $\Omega_u$ changes}
 \State {\bf move} message (if any) with id $\Omega_u$ to front of tree queue
\EndProcedure
  \end{algorithmic}
\end{algorithm}
\begin{algorithm}[H]
  \caption{Broadcast Service (for node $u$)}
  \label{alg:paxos:bcast}
  \begin{algorithmic}[1]
  \scriptsize
  \While{true}
  	\State {\bf wait} for at least one queue from 
	\State $\ \ \ $ $\{tree, leader, change\}$ to become non-empty
	\State {\bf dequeue} a message from each non-empty queue and
	\State $\ \ \ $ combine into one message $m$.
	\State {\bf broadcast}$(m)$ then wait for $ack$
  \EndWhile
   \end{algorithmic}
\end{algorithm}

\end{minipage}
\caption{Support services used by wPAXOS.  Notice, the broadcast service schedules message
broadcasts from the queues maintained by the other three services. We assume that when a message is received
it is deconstructed into its constituent service messages which are then 
passed to the {\tt receive} procedure
of the relevant service. This basic receive logic is omitted in the above.}
\label{fig:paxos}
\end{figure}



\subsubsection{Algorithm}

Our algorithmic strategy is to implement the logic of the classic PAXOS asynchronous agreement
algorithm~\cite{paxos,paxos-simple} in our abstract MAC layer model.
To do so, we implement and analyze a collection {\em services}
that can be connected to the high-level PAXOS logic to run it in our model.
These services are responsible for disseminating proposer messages and acceptor
responses (in an efficient manner), as well as notifying the high level PAXOS logic 
when to start over with a new proposal number. They also provide the leader election
service needed for liveness.
As mentioned, we call this combination of the high-level PAXOS logic with our  model-specific
support services, {\em wireless PAXOS}  (wPAXOS). 

We note, that if we did not care about time complexity,
our services could disseminate messages and responses with simple flooding services.
To achieve an optimal $O(D\cack)$ time, however,
requires a more complicated strategy.
In more detail, our services build, in a distributed fashion,
an eventually stabilized shortest path tree rooted at the eventual leader in the network.
Once stabilized, acceptors can efficiently send their responses to the leader
by aggregating common response types as they are routed up the tree.

We proceed below by first describing these support services, and then describing how to connect them
to the standard PAXOS logic. We conclude by proving our needed safety and liveness properties
for the resulting combined algorithm. 
In the following, we assume the reader is already familiar with the PAXOS algorithm (if not, see~\cite{paxos-simple}),
and will focus on the new service algorithms specific to our model it uses and how it uses them.


\paragraph{Services.}
Our wPAXOS algorithm requires the four support services (see Figure~\ref{fig:paxos} for the pseudocode).
The first three, {\em leader election}, {\em change}, and {\em tree building} each maintain
a message queue. The fourth, {\em broadcast}, is a straightforward loop that takes messages
from the front of these queues, combines them, then broadcasts the combined message---allowing the algorithm to multiplex multiple
service on the same channel. In the following, therefore, we talk about the messages
of the leader election, change, and tree building services as if they were using a dedicated channel.

{\em Leader Election.} This service maintains a local variable $\Omega_u$, containing an id of some
node in the network, at each node $u$.
The goal of this service is to eventually stabilize these variables to the same id network-wide.
It is implemented with standard flooding logic.

{\em Tree Building.} This service attempts to maintain in the network, for each node $v$, a shortest-path tree rooted at $v$.
The protocol runs a Bellman-Ford style iterative refinement procedure for establishing these trees,
with the important optimization that the messages corresponding to the current leader get priority in each node's tree service queue.
This optimization ensures that soon after the network stabilizes to a single leader a tree rooted at that leader is completed.
%

{\em Change.} This service is responsible for notifying PAXOS proposers when they should start a new proposal number.
The goal for this service is to ensure that the eventual leader $u_{\ell}$ starts a new proposal after the other
two services have stabilized: that is, after the leader election service has stabilized to $u_{\ell}$ across the network,
and the tree rooted at $u_{\ell}$ is complete. Proposals generated after this point, we will show, 
are efficiently processed. The service must also guarantee, however, not to generate {\em too many} changes
after this point, as each new proposal can delay termination.


\paragraph{Connecting PAXOS Logic to Support Services.}
We now describe how to combine the standard PAXOS logic 
with our model-specific support services
to form the wPAXOS algorithm.
Recall, in PAXOS there are two roles: {\em proposers} and {\em acceptors}.\footnote{There is sometimes a third role
considered, called {\em learners}, that are responsible for determining when a decision has been made. As is common,
however, we collapse this role with the role of the proposer. When a proposer determines it can can decide a value---i.e.,
because a proposal for this value was accepted by a majority of acceptors---it does so and floods this decision to the rest
of the network.} In wPAXOS all nodes play both roles. We describe below how both roles interact with our services.

{\em Proposers.}
Proposers generate {\em prepare} and {\em propose} messages (the latter sometimes called {\em accept} messages),
associated with a given proposal number.
A proposal number is a $tag$ (initially $0$) and the node's id. The pairs are compared lexicographically.
A key detail in any PAXOS deployment is the conditions under which a proposer chooses a new proposal number
and starts over the proposal process.
In wPAXOS, a proposer starts over when its change service locally calls
{\em GenerateNewPAXOSProposal}(). At this point it increases its $tag$ to be $1$ larger than the largest
tag it has previously seen or used.
If this new proposal number is rejected at either stage,
{\em and} the proposer learns of a larger proposal number in the system during this process (i.e., because an acceptor
appended a previous message to its rejection), {\em and} it is still the leader according to its leader election service,
it increases its tag number and attempts a new proposal.
If this new proposal also fails, it waits for the change service before trying again.
In other words, a proposer will only try up to $2$ proposal numbers for each time it is notified by the change service
to generate a new proposal.

To disseminate the {\em prepare} and {\em propose} messages
generated by proposers we assume a simple flooding algorithm (due to its simplicity, we do not show
this pseudocode in Figure~\ref{fig:paxos}): if you see a proposer message from $u$
for the first time, add it to your queue of proposer messages to rebroadcast.
To prevent old proposers messages from delaying new proposer messages we have each proposer maintain
the following invariant regarding the message queues used in this flooding:
At all times, the queue is updated to ensure:
 (1) only contains messages from the current leader; and (2) only contains messages associated with the largest
proposal number seen so far from that leader.

{\em Acceptors.}
We now consider the acceptors.
%
When acceptor $v$  generates a response to a {\em prepare} or {\em propose} message from proposer $u$,
$v$ labels the response with the destination $parent[\id{u}]$ then adds it to its acceptor broadcast queue.
This message is then  treated like a unicast message: even though it will be broadcast (as this is the only communication primitive
provided by our model), 
it will be ignored by any node except $parent[\id{u}]$.
Of course, if $v=u$ then the message can skip the queue and simply be passed over to the proposer logic.

To gain efficiency, we have acceptors aggregate messages in their acceptor queue when possible.
In more detail, if at any point an acceptor $v$ has in its queue multiple responses of the same type (positive or negative)
to the same proposer message, to be sent to same $parent$: it can combine them into a single {\em aggregated} message.
This single message retains the the type of responses as well as the proposal number the responses reference,
but replaces the individual responses with a count.
We can combine aggregated messages with other aggregated messages and/or non-aggregated messages in the same way
(i.e., given two aggregated messages with counts $k_1$ and $k_2$, respectively, we can combine them into an aggregated message with 
count $k_1 + k_2$).

Notice, PAXOS sometimes has acceptors respond to a {\em prepare} message with a previous proposal (i.e., combination of a proposal
number and value). When aggregating multiple messages of this type (i.e., containing previous proposal) we keep
only the previous proposal with the largest proposal number.
We also assume PAXOS implements the standard optimization that has acceptors, when rejecting a message,
append the larger proposal number to which they are currently committed.
We aggregate these previous proposals in the same way we do with positive responses to prepare messages---by maintaining
only the largest among those in the messages we are aggregating.

Finally, as with the proposers, the acceptors keep their message queue up to date by maintaining the invariant
that at all times, the queue is updated to ensure:
 (1) only contains messages in response to propositions from the current leader; and (2) only contains responses associated with the largest
proposal number seen so far from that leader.

{\em Deciding.}
When a proposer learns it can decide a value $val$ (i.e., because it has learned that
at least a majority of acceptors replied with {\em accept} to a proposal containing $val$),
it will decide $val$ then flood a {\em decide}$(val)$ message.
On receiving this message, a node will decide $val$.

\subsubsection{Analysis}
\label{sec:paxos:analysis}

We now prove that our wPAXOS algorithm solves consensus in $O(D\cack)$ time---matching the relevant lower bounds. 
%
In the following, let a {\em message step} (often abbreviated below as just ``step") be an event in an execution
where a message or ack is received.
Notice, because we consider deterministic algorithms and assume that local computation does not require any time,
the behavior of an algorithm in our model is entirely describe by a sequence of message steps. 
%
Let a {\em proposition} be a combination of a proposer, a proposer message type (i.e., either {\em prepare} or {\em propose}),
and a proposal number. 
For a fixed execution of wPAXOS:
let $a(p)$, for a given proposition $p$, be the number acceptors that generate an affirmative response to proposition $p$;
and let $c(p)$ be the total count of affirmative responses for proposition $p$ received by the originator of $p$.
Similarly, let $c(p,s)$, for step $s$, be the total count of affirmative responses to $p$ received by the end of step $s$.


\paragraph{Safety.}
In the standard asynchronous network model, PAXOS takes for granted that proposers properly count 
the relevant
responses to their propositions, as the model reliably delivers each response, labeled with the acceptor that generated it. 
In wPAXOS, however, we aggregate responses using shortest-path trees
that might change during throughout the execution.
Before we can leverage the standard safety arguments for PAXOS,
we must first show that this aggregation in dynamic trees does not compromise the integrity of the proposer response counts.

\begin{lemma}
Fix some execution of wPAXOS. Let $p$ be a proposition generated by $u$ in this execution.
It follows that $c(p) \leq a(p)$.
\label{lem:paxos:1}
\end{lemma}
\begin{proof}
In this proof, for a fixed execution, proposition $p$, node $v$, and step $s$,
we
define $q(p,v,s)$ be the sum of the following values:
(1) the total count of affirmative responses to proposition $p$ 
in node $v$'s acceptor message queue after step $s$;
(2) the total count of the affirmative responses to $p$ in the message $v$ is in the process of sending to some $w$
after step $s$, assuming that $w$ has not yet received the message at this point;
(3) $1$, if the acceptor at $v$ will generate an affirmative response to $p$ at some later step $s' > s$ in the execution.
Similarly, let $Q(p,s) = \sum_{v\in V} q(p,v,s)$. 
Fix some proposition $p$ with originator $u$. 
To prove our lemma it is sufficient
to prove the following invariant:
\begin{center}
{\em (*) For every step $s\geq 0$:  $Q(p,s) + c(p,s)  \leq a(p)$.}
\end{center}

We  prove (*) by induction on $s$.
The {\em basis} ($s=0$) of this induction is trivial:
by definition $Q(p,0) = a(p)$ and $c(p,0) = 0$. 
For the {\em inductive step} we assume (*) holds through step $s$.
There are three cases for step $s+1$ that are relevant to our hypothesis.

{\em Case $1$:} Assume $s+1$ has some acceptor $v\neq u$ receive a message that causes it to 
generate a response to $p$.
This action transfers quantity $1$ from $q(p,v,s)$ (contributed from element (3) of the definition of $q$)
and transfers it either to the messages in $v$'s queue, or, if it is not in the process of sending a message
when $s+1$ occurs, into a message being sent by $v$. In either case, this value of $1$ now
shows up in $q(p,v,s+1)$ either from element (1) {\em or} (2) of its definition/
It follows that $Q(p,s+1) = Q(p,s)$ and $c(p,s+1) = c(p,s)$.

{\em Case $2$:}
Assume $s+1$ has some acceptor $w\neq u$ receive a message $m$ from $v$
that contains affirmative responses
to $p$ {\em and} $v$ addressed the message to $w$.
In this case, the count of affirmative responses
contained in this message was moved unchanged from
$q(p,v,s)$ to $q(p,w,s+1)$. Once again $Q(p,s) = Q(p,s+1)$ and $c(p,s+1)=c(p,s)$.

{\em Case $3$:}  Assume step $s$ has $u$ receive a message 
with $k$ affirmative responses to $p$ that is addressed to $u$.
In this case, $u$ adds $k$ to its count and discards the message.
It follows that $Q(p,s+1) = Q(p,s) - k$ and $c(p,s+1) = c(p,s) + k$.
\end{proof}

The below lemma comes from~\cite{paxos-simple}, where Lamport shows
that proving this lemma provides agreement.
If we can prove this below lemma, in other words, then we can apply
the arguments of~\cite{paxos-simple} to establish agreement.
Notice, this proof leverages Lemma~\ref{lem:paxos:1}.

\begin{lemma}
Fix an execution of wPAXOS. 
Assume proposer $u$ generates a proposal with value $val$ and proposal number $x$.
It follows that there exists a set $S$ consisting of a majority of acceptors such
that either: (a) no acceptor in $S$ has accepted a proposal with  number smaller than $x$,
or (b) $val$ is the value of the highest-numbered proposal among
all proposals with numbers less than $x$ accepted by the acceptors in $S$.
\label{lem:paxos:2}
\end{lemma}
\begin{proof}
Fix some execution $\alpha$. Let $p'$ be a proposal proposition generated by $u$ for proposal number $x$.
Let $p$ be the prepare proposition from $u$ with this same number $x$ that must, by the definition of the algorithm,
 precede $p'$.
Let $A$ be the set of acceptors that commits to $p$.
Let $s$ be the step at which $c(p,s)$ becomes greater than $n/2$ for the first time.
%
We define a directed graph $G(\alpha,p,s) = (V,E')$ as follows.
Add $(w,v)$ to $E'$ iff at some step $s' \leq s$ in $\alpha$, acceptor $w$ sent a message about a positive
response to $p$ to $v$. 
Let $A_u \subset A$ be the subset of $A$ that have a path to $u$ in $G(\alpha,p,s)$.

Consider the case where at least one node in $A_u$ had previously committed to a proposal number $\hat x$ when $p$ arrived.
We know $\hat x < x$ as this node subsequently committed to $x$.
Let $x_{max} \geq \hat x$ be the largest such proposal number of a previous commitment and $val_{max}$ be the corresponding value.
Let $P$ be the set of previous proposals that $u$ receives in response to $p$.
We argue that $(x_{max}, val_{max})$ is in $P$ and has the largest proposal number of any proposal in $P$.
If this is true, it follows that $u$ will adopt $val_{max}$ for its proposal $p'$, as needed by the lemma.

To show this is true,
let $v_{max}$ be an acceptor in $A_u$ that previously committed to that proposal.
By the definition of $G(\alpha,p,s)$, there is a casual path of messages on which this 
proposal could travel to $u$.
The only way it could be discarded before arriving at $u$
is if at some point a larger proposal in response to $p$ joins
the path from $v_{max}$ to $u$, {\em before} the $x_{max}$ proposal has been sent forward. 
By assumption, however, no process in $A_u$ has a larger proposal number.
It follows that this larger proposal must have originated from some $v' \notin A_u$.
However, if this previous proposal can intercept the path from $v_{max}$ to $u$ before $v_{max}$'s messages
have passed that point, then, by the definition of $G(\alpha,p,s)$.
$u$ must have a path to $u$ in $G(\alpha,p,s)$, and therefore must be in $A_u$.

We have now shown that properties (a) and (b) hold's for $A_u$. 
To prove the lemma, we are left to show that $A_u$ must hold a majority of acceptors.
Assume for the sake of contradiction that $|A_u| \leq n/2$. 
At step $s$, however, proposer $u$ has $c(p,s) > n/2$.
By the definition of $G(\alpha,p,s)$, 
there is no causal path of messages being sent and received
from nodes in $V\setminus A_u$ to $u$ that arrive at $u$ before
it counts a majority of acceptances.
Therefore, the execution $\alpha'$ in which {\em no} node in $V \setminus A_u$
commits to $p$ is indistinguishable from $\alpha$ with respect to $u$ through step $s$.
In $\alpha'$, however, $|A| \leq n/2$ and therefore $a(p) \leq n/2$.
In this same execution, we also know $c(p) > n/2$. The result: in $\alpha'$, $c(p) > a(p)$.
This contradicts Lemma~\ref{lem:paxos:1}
\end{proof}

Part of the difficulty tackled by wPAXOS is efficiently disseminating responses form acceptors even though the
messages are of bounded size (i.e., can only hold $O(1)$ ids). If message size was instead unbounded,
we could simply flood every response we have seen so far.
Given this restriction, however, we must make sure that the proposal numbers used by proposers
in wPAXOS do not grow too large. In particular, we show below that these proposal numbers always
remain small enough to be represented by $O(\log{n})$ bits (messages must at least this
large by the assumption that they can hold a constant number of ids in a network of size $n$).

\begin{lemma}
There exists a constant $k$ such that the tags used in proposal numbers of wPAXOS
are bounded by $O(n^k)$.
\label{lem:paxos:3}
\end{lemma}
\begin{proof}
Each node $u$ can only locally observe $O(n^2)$ change events.
its leader variable can only have $n$ possible values and these values
strictly increase. And for each $v$,  $dist[v]$ at $u$ can only have $n$ values it strictly decreases.
Each change event can lead to at most $2$ new proposal numbers at each node in the network
(recall: the algorithm restricts a proposer to attempting at most $2$ new numbers in response to a given
call of {\em GenerateNewPAXOSProposal}). We also note that there are $n$ total nodes and the tags
in proposal numbers are increased to be only $1$ larger than the largest tag previously seen.
A straightforward induction argument bounds the largest possible tag size as polynomial in $n$ as needed.
\end{proof}

\paragraph{Liveness.}
 We now establish that all nodes in wPAXOS will decide by time $O(D\cack)$. The main idea
 in this proof is to consider the behavior of the algorithm after the tree and leader election service
 stabilize. WE show this stabilization occurs in $O(D\cack)$ time and after this point the leader
 will continue to generate proposals and end up reaching a decision within an additional
 $O(D\cack)$ time.
 
 \begin{lemma}
 Fix some execution of wPAXOS. 
 Every node decides in $O(D\cack)$ time.
 \label{lem:paxos:4}
 \end{lemma}
 \begin{proof}
Let $\hat s$ be the step that generates the {\em final} change in this execution.
Let $t(\hat s)$ be the global time at which this step occurs.
Repurposing a term from the study of partially synchronous model,
we call $t(\hat s)$ the {\em global stabilization time} (GST) as it defines a time by which point:
(a) the whole network has stabilized to the same leader, $\ell$;
(b) the tree defined by $parent[\ell]$ pointers in the network has stabilized
to a shortest-path spanning tree rooted at $\ell$.

Consider $\ell$'s behavior after the GST.
We know that $\hat s$ generates a change message.
Because this change message has a timestamp as larger (or larger) than other change message
in the execution, 
every node will receive a change message with this timestamp for the first time
somewhere in the 
$[t(\hat s), t(\hat s) + O(D\cack)]$. This is the last change message any node will pass on to their {\em UpdateQ} procedure
in the change service.
This is significant, because it tells us that $\ell$ will generate a new proposition at some on or after the GST, but not {\em too much} after.
And, this is the last time the change service will cause $\ell$ to generate a proposition.

We now bound the efficiency of $\ell$'s propositions after the GST.
Notice, after GST only messages from $\ell$ can be stored in proposer queues.
When $\ell$ generates a new proposition message with a larger proposal number than any previous such messages,
it will necessarily propagate quickly---reaching a major it of acceptors in $O(D\cack)$ time.
We can show that acceptor responses return to $\ell$ with the same bound,
as every such response is at most distance $D$ from $\ell$ in the shortest path tree,
and can be delayed by at most $O(\ack)$ time at each hop (if multiple responses arrive
at the same node they will simply be aggregated into one response)

Now we consider the fate of $\ell$'s propositions after GST.
We know from above that it generates a new proposition with a new proposal number on or after GST.
If receives enough commits to its prepare message for this proposal number,
then it follows it will go on to gather enough accepts to decide (as no other proposer,
after GST, can have its proposals considered).
If it fails to gather enough commits, it will move on to the next proposal after counting a majority
of the acceptors rejecting its prepare. By the algorithm, however, it will learn the largest proposal
number that one of this set has previously committed to. Its subsequent proposal number
will be larger, and therefore this same majority will end up sending him enough commits to
move on toward decision.
It follows that a constant number of propositions is sufficient for $\ell$ to decide.

To tie up the final loose ends, we note that once $\ell$ decides, all nodes decide within an additional
$O(D\cack)$ time, and that GST is bounded by $O(D\ack)$ as well: 
it takes this long for the leader election service to stabilize, after which the tree rooted at the leader
must be stabilized within in addition $O(D\cack)$ time.
 \end{proof}

 \paragraph{Pulling Together the Pieces.}
 We are now ready to combine our key lemmas to prove the final correctness of wPAXOS.
 
 \begin{theorem}
 The wPAXOS algorithm solves consensus in $O(D\cack)$ time in our abstract MAC layer
 model in any connected network topology, where $D$ is the network diameter and nodes have unique
 ids and knowledge of network size.
 \label{thm:paxos}
 \end{theorem}
 \begin{proof}
 The {\em validity} property is trivial. The {\em termination} property as well as the $O(D\cack)$ bound on termination
 follow directly from Lemma~\ref{lem:paxos:4}.
 To satisfy agreement we combine Lemma~\ref{lem:paxos:2} with the standard argument of~\cite{paxos-simple}.
 \end{proof}

 \section{Conclusion}
Consensus is a key primitive in designing reliable distributed systems.
Motivated by this reality and the increasing interest in wireless distributed systems,
in this paper we studied the consensus problem in a wireless setting.
In particular, we proved new upper and lower bounds for consensus
in an abstract MAC layer model---decreasing the gap between our
results and real deployment. We prove that (deterministic) 
consensus is impossible with crash failures,
and that without crash failures, it requires unique ids and knowledge of the
network size. We also establish a lower bound on the time complexity of any consensus solution.
We then present two new consensus algorithms---one optimized for single hop networks and one
that works in general multihop (connected) networks.

In terms of future work, there are three clear next steps that would help advance
this research direction. The first is to consider consensus in an abstract MAC layer model
that includes unreliable links in addition to reliable links.
The second is to consider what additional formalisms might allow deterministic consensus
solutions to circumvent the impossibility concerning crash failures. 
In the classical distributed systems setting, failure detectors were used for this purpose.
In the wireless world, where, among other things, we do not always assume {\em a priori} knowledge
of the participants, this might not be the most natural formalism to deploy.
The third direction is to consider randomized algorithms, which might provide better performance
and the possibility of circumventing our crash failure, unique id, and/or network size knowledge lower bounds.

\bibliographystyle{plain}
\bibliography{wireless-consensus,wireless,sinr}


\end{document}